\documentclass[11pt,twocolumn]{article}
\usepackage{latexsym}
\usepackage{amsmath}
\usepackage{amssymb}
\usepackage{amsthm}
\usepackage{epsfig}
\usepackage{natbib}
\newtheorem{theorem}{Theorem}[section]

\newtheorem{lemma}[theorem]{Lemma}

\newtheorem{definition}[theorem]{Definition}

\usepackage{url}

\makeatletter
\def\url@leostyle{%
  \@ifundefined{selectfont}{\def\UrlFont{\sf}}{\def\UrlFont{\small\ttfamily}}}
\makeatother
\usepackage[compact]{titlesec} 
\usepackage{mdwlist}
\newcommand{\comment}[1]{}
\begin{document}
\title{Parimutuel Betting on Permutations}

\author{
Shipra Agrawal \thanks{Department of Computer Science, Stanford University. 
%Stanford, CA 94305, 
Email: shipra@stanford.edu}
    \and 
Zizhuo Wang \thanks{Department of Management Science and Engineering, Stanford University. %Stanford, CA 94305, 
Email: zzwang@stanford.edu}
%Dept. of Mgmt. Sci. and Engg.\\
%Stanford University\\
%Stanford, CA 94305\\
%	zzwang@stanford.edu\\
 \and
Yinyu Ye \thanks{Department of Management Science and Engineering, and, Institute of Computational and Mathematical Engineering (ICME), Stanford University. Email: yinyu-ye@stanford.edu} \thanks{Research supported in part by NSF DMS-0604513.}
%Dept. of Mgmt. Sci. and Engg.\\
%and, by courtesy,\\
%Electrical Engineering\\
%Stanford University\\
%Stanford, CA 94305\\
%yinyu-ye@stanford.edu\\
}
\date{ }
\maketitle
\begin{abstract}
{\small
We focus on a permutation betting market under parimutuel call auction model where traders bet on the final ranking of $n$ candidates.
We present a {\it Proportional Betting} mechanism for this market.
Our mechanism allows the traders to bet on any subset of the $n^2$ `candidate-rank' pairs, and rewards them proportionally to the number of pairs that appear in the final outcome. 
%Subset betting proposed in \cite{nikolova} forms a special case of this mechanism. 
We show that market organizer's decision problem for this mechanism can be formulated as a convex program of polynomial size.
%with only $O(n^2)$ constraints.
More importantly, the formulation yields a set of $n^2$ {\it unique} marginal prices that are sufficient to price the bets in this mechanism, and are computable in polynomial-time. The marginal prices reflect the traders' beliefs about the marginal distributions over outcomes. 
We also propose techniques to compute the joint distribution over $n!$ permutations from these marginal distributions.
%We propose techniques to compute the complete price vector over $n!$ permutations from these marginal prices, which can be interpreted as the traders' belief about joint distribution over outcomes. 
We show that using a maximum entropy criterion, we can obtain a concise parametric form (with only $n^2$ parameters) for the joint distribution which is defined over an exponentially large state space. We then present an approximation algorithm for computing the parameters of this distribution. In fact, the algorithm addresses the generic problem of finding the maximum entropy distribution over permutations that has a given mean, and may be of independent interest.
% over permutations may be of independent interest and could find applications elsewhere.
}
%The marginal prices in this mechanism reflect the traders' belief about marginal distributions over candidates and positions in the outcome. 
%The marginal prices can be interpreted as market generated marginal distributions over `positions for a given candidate' or `candidates for a given position', and are shown to satisfy certain price-consistency constraints proposed by \cite{Lange} for parimutuel call auction models. 
%To make the mechanism more useful as an information aggregation device, we propose techniques for obtaining the joint distribution over all $n!$ outcome permutations from these marginal distributions. Using a maximum-entropy criteria, we obtain a joint distribution that has concise parametric representation in terms of only $n^2$ parameters; and give an approximation algorithm to compute these parameters. 
%to arbitrary accuracy in polynomial time.  
%We believe this is the first result on {\it pricing} a parimutuel call auction under exponentially large outcome space.
%and an intuitive bayesian interpretation. The polynomial-time approximation scheme that  for computing this maximum entropy distribution over permutations may be of independent interest and could find applications elsewhere. 
% The ideas we present on approximating the maximum entropy distribution over permutations may be of independent interest and could find applications elsewhere.
\end{abstract}

\section{Introduction}
%A common task is how to induce people to collect information on chosen topics, and to reveal what they know to each other, in order to produce consensus estimates on those topics. If people were always open and honest with themselves and others, and believed that others were this way as well, the ideal institution for this task would probably be simple conversation. Unfortunately, this ideal rarely holds. Information markets are institutions for accomplishing this task via trading in speculative markets, at least on topics where truth can be determined later.  
Prediction markets are increasingly used as an information aggregation device in academic research and public policy discussions. The fact that traders must ``put their money where their mouth is" when they say things via markets helps to collect information. To take full advantage of this feature, however, we should ask markets the questions that would most inform our decisions, and encourage traders to say as many kinds of things as possible, so that a big picture can emerge from many pieces. Combinatorial betting markets hold great promise on this front. Here, the prices of contracts tied to the events have been shown to reflect the traders' belief about the probability of events.
%a market generated probability estimate of the likelihood of events. 
%(belief). 
Thus, the pricing or ranking of possible outcomes in a combinatorial market is an important research topic. 

%In this paper, we consider combinatorial betting In this paper, 
We consider a permutation betting scenario where traders submit bids on final ranking of $n$ candidates, for example, an election or a horse race.
%under parimutuel call auction model proposed by Lange and Economides \cite{Lange} and studied by Ye et al \cite{Ye-CPCAM} for non-combinatorial setting. In this model he traders submit bids on possible states, for example, on outcome of an election, or a horse race. 
%The auctioneer will collect all the bids, close the market and then determine which bids to accept in order to maximize his worst case profit. This is a parimutuel call auction model studied by Lange and Economides \cite{Lange}.
%In permutation betting scenario considered here, 
The possible outcomes are the $n!$ possible orderings among the candidates, and hence there are $2^{n!}$ subset of events to bid on. 
In order to aggregate information about the probability distribution over the entire outcome space, one would like to allow bets on all these event combinations. 
%However, such betting mechanisms often result in intractable market maker and pricing problems. Further, they exacerbate the thin market problems by dividing participants' attention among an exponential number of outcomes \cite{chen-short, hanson03}.
However, such betting mechanisms are not only intractable, but also exacerbate the thin market problems by dividing participants' attention among an exponential number of outcomes \cite{chen-short, hanson03}.
%, making
%the likelihood of finding agreeable trades low even with multilateral matching. This
%may cause a thin market problem and ultimately a failure of information aggregation
%, that is, substantial trading activity happens in a certain limited number of assets
% and irrational participation problems that plague standard information markets.
Thus, there is a need for betting languages or mechanisms that could restrict the possible bid types to a tractable subset and at the same time provide substantial information about the traders' beliefs. %joint distributions. 

\comment{In this paper, we propose a ``proportional betting" mechanism that can be formulated as a convex program of polynomial size. Further, we show that our formulation yields a set of $n^2$ ``unique marginal prices" that can be used to construct meaningful joint distribution over the outcome space.}
%In this mechanism, the traders bet on one or more of the $n^2$ `candidate-position' pairs, and receive rewards  proportional to the number of pairs that appear in the final outcome. For example, a trader may bid on ``Horse A will finish in position 2 OR Horse B will finish in position 4 OR....".  She receives a reward of 2\$ if both Horse A \& Horse B finish at the specified positions 2 \& 4 respectively; and a reward of 1\$ if only one horse finishes at the position specified, etc. Subset betting introduced by \cite{nikolova} forms a special case of this betting language.

%We show that auctioneer problem for this mechanism can be formulated as a linear program of polynomial size. Further, the formulations yields a set of $n^2$ ``unique marginal prices" computable in polynomial-time. These prices can be interpreted as marginal distribution across ``all rankings for a given candidate" or ``all candidates for a given position", and are shown to satisfy the ``price-consistency constraints" proposed by \cite{Lange}. We also propose ways to use these marginal distributions to compute a joint distribution over all $n!$ outcomes that has a concise representation and an intuitive bayesian interpretation. 

\subsection{Previous Work}
%Nikolva Betting on Permutations, SUbset betting
%LMSR: Combinatorial Information Market Design, Robin Hanson
%Pennock LMSR pricing, Complexity of Combinatorial Market Makers* 
%PRICING COMBINATORIAL MARKETS FOR TOURNAMENTS, YILING CHEN, SHARAD GOEL AND DAVID PENNOCK
Previous work on parimutuel combinatorial markets can be categorized under two types of mechanisms:
a) posted price mechanisms including the Logarithmic Market
Scoring Rule (LMSR) of Hanson \cite{hanson03, hanson07} and the Dynamic Pari-mutuel Market-Maker (DPM) of Pennock \cite{pennock-DPM}
b) call auction models developed by Lange and Economides \cite{Lange}, Peters et al. \cite{Ye-CPCAM},
%, convex formulation by Peters et al. \cite{Ye-CPCAM}) 
%in which the organizer collects all orders, closes the market and then determines which orders to accept and which to reject.
in which all the orders are collected and processed together at once. 
An extension of the call auction mechanism to a dynamic setting similar to the posted price mechanisms, and a comparison between these models can be found in Peters et al. \cite{Ye-comparison}. %, and provide a comparison between these models. 
%where the market organizer will make immediate,binding decisions about orders. A comparison between these models can also be found in \cite{Ye-comparison}.

Chen et al. (2008) \cite{chen08} analyze the computational complexity of market maker
pricing algorithms for combinatorial prediction markets under LMSR model.
They examine both permutation combinatorics,
where outcomes are permutations of objects, and Boolean
combinatorics, where outcomes are combinations of binary
events. Even with severely limited languages,
they find that LMSR pricing is \#P-hard, even when
the same language admits polynomial-time matching without
the market maker. Chen, Goel, and Pennock
\cite{goel08} study a special case of Boolean combinatorics and
%in which participants bet on how far a team will advance
%in a single elimination tournament, for example a sports playoff like the NCAA college basketball tournament. They
provide a polynomial-time algorithm for LMSR pricing in
this setting based on a Bayesian network representation of
prices. They also show that LMSR pricing is NP-hard for a
more general bidding language.

More closely related to our work are the studies by Fortnow et al. \cite{fortnow} and 
Chen et al. (2006) \cite{nikolova} on {\it call auction} combinatorial betting markets. 
Fortnow et al. \cite{fortnow} study the computational complexity of
finding acceptable trades among a set of bids in a Boolean
combinatorial market. 
Chen et al. (2006) \cite{nikolova} analyze the auctioneer's matching problem
for betting on permutations, examining two bidding
languages. Subset bets are bets of the form candidate $i$
finishes in positions $x$, $y$, or $z$ or candidate $i$, $j$, or $k$ finishes
in position $x$. Pair bets are of the form candidate $i$
beats candidate $j$. They give a polynomial-time algorithm
for matching divisible subset bets, but show that matching
pair bets is NP-hard.

\subsection{Our Contribution}
This paper extends the above-mentioned work in a variety of ways. We propose a more generalized betting language called {\it Proportional Betting} that encompasses {\it Subset Betting} \cite{nikolova} as a special case. 
Further, we believe that ours is the first result on {\it pricing} a parimutuel call auction under permutation betting scenario.
%We propose a `proportional betting mechanism' for permutation betting under call auction model. 

In our proportional betting mechanism, the traders bet on one or more of the $n^2$ `candidate-position' pairs, and receive rewards  proportional to the number of pairs that appear in the final outcome. For example, a trader may place an order of the form ``Horse A will finish in position 2 OR Horse B will finish in position 4".  He \footnote{`he' shall stand for `he or she'} will receive a reward of \$2 if both Horse A \& Horse B finish at the specified positions 2 \& 4 respectively; and a reward of \$1 if only one horse finishes at the position specified. The market organizer collects all the orders and then decides which orders to accept in order to maximize his worst case profit.
%Subset betting introduced by \cite{nikolova} forms a special case of this betting language.
%In this paper, we will describe a new mechanism for centrally organizing a permutation betting market to maximize liquidity. Our model is called the Proportional Betting Mechanism. In this mechanism, traders are allowed to bid on any subset of candidate-rank pairs, and receive reward proportional to the number of wins. 

In particular, we will present the following results:
\begin{itemize*}
\item We show that the market organizer's decision problem for this mechanism can be formulated as a convex program with only $O(n^2+m)$ variables and constraints, where $m$ is the number of bidders. 
%\item Our mechanism is parimutuel in the sense that the payouts made by the auctioneer will be completely funded by accepted orders.
\item We show that we can obtain, in polynomial-time, a small set ($n^2$) of `marginal prices' that satisfy the desired price consistency constraints, and are sufficient to price the bets in this mechanism. 
%and are parimutuel in the sense that the payouts made to the bidders will be completely funded by the prices charged for the accepted orders.
\item We show that by introducing non-zero starting orders, our mechanism will produce unique marginal prices. %While we require non-zero starting orders to generate the unique marginal prices, we prove that 
%These prices converges to a unique limit as we drive the magnitude of starting orders uniformly to zero. Moreover, this limit can be computed to any degree of accuracy in polynomial-time.
%Moreover, such a price vector can be found using the aforementioned algorithm.
\item We suggest a maximum entropy criteria to obtain a maximum-entropy joint distribution over the $n!$ outcomes from the marginal prices. Although defined over an exponential space, this distribution has a concise parametric form involving only $n^2$ parameters. 
%which are shown to be computable to any degree of accuracy in polynomial time. Moreover, we show that the maximum-entropy distribution 
Moreover, it is shown to agree with the maximum-likelihood distribution when prices are interpreted as observed statistics from the traders' beliefs.
\item We present an approximation algorithm to compute the parameters of the maximum entropy joint distribution to any given accuracy in (pseudo)-polynomial time \footnote{The approximation factors and running time will be established precisely in the text.}. In fact, this algorithm can be directly applied to the generic problem of finding the maximum entropy distribution over permutations that has a given expected value, and may be of independent interest.
\end{itemize*}
\section{Background}
%\section{Parimutuel call auction model}
\label{background}
%\subsection{Convex Parimutuel Call Auction Model (CPCAM)}
%\paragraph{Parimutuel Market Microstructure (PMM)}
%Our betting model is derived from parimutuel market model developed by Lange and Economides \cite{Lange} for contingent claims market that is run by a call auction. 
%In this model, the market organizer guarantees a fixed payout if an order is accepted and one of its specified states is realized. The traders submit bids to the organizer which specify the states they want contingent claims over, the price they are willing to pay for the bid, and the number of identical bids that they will buy. The organizer will then decide whether to accept or reject each bid. If the bid is accepted, the organizer also decides the number of bids to accept and the price per bid to be collected from the bidder. The market is not run as a continuous auction. Instead, the market organizer will collect all bids, close the market and then announce the accepted bids, quantities and prices.  This model has been implemented by Goldman Sachs and Deutsche Bank to run markets on options for economic data \cite{cass02}.

%Can be formulated as CPCAM formulation
%\paragraph{Convex Parimutuel call auction Model (CPCAM)} 
%Peters et al. \cite{Ye-CPCAM} gave a alternative formulation which have similar constraints as the PMM model discussed above, but is also a convex program. 
In this section, we briefly describe the Convex Parimutuel Call Auction Model (CPCAM)  developed by Peters et al. \cite{Ye-CPCAM} that will form the basis of our betting mechanism.

Consider a market with one organizer and $m$ traders or bidders.
There are $S$ states of the world in the future on which the traders are submitting bids. For each bid that is accepted by the organizer and
contains the realized future state, the organizer will pay the
bidder some fixed amount of money, which is assumed to be $\$1$ without loss of generality. The organizer collects all the bids and decides which bids to accept in order to maximize his worst case profit.

Let $a_{ik} \in \{0,1\}$ denote the trader $k$'s bid for state $i$. Let $q_{k}$ and $\pi_k$ denote the limit quantity and limit price for trader $k$, i.e., trader $k$'s maximum number of orders requested and maximum price for the bid, respectively. The number of orders accepted for trader $k$ is denoted by $x_k$, and $p_i$ denotes the price computed for outcome state $i$. $x_k$ is allowed to take fractional values, that is, the orders are `divisible' in the terminology of \cite{nikolova}.
Below is the convex formulation of the market organizer's problem given by \cite{Ye-CPCAM}:
\begin{equation}\label{CPCAM}
\begin{array}{lll}
\displaystyle\max_{x,s,r} & \pi^T x-r + \mu \sum_{i=1}^S \theta_i \log(s_i) & \\
\mbox{s. t.} & \sum_k a_{ik}x_k + s_i = r & 1\le i\le S\\
& 0 \le x \le q  & \\
& s\ge 0&
\end{array}
\end{equation}
The `parimutuel' price vector $\{p_i\}_{i=1}^S$ is given by the dual variables associated with the first set of constraints. The parimutuel property implies that when the bidders are charged a price of $\{\sum_i a_{ik}p_i\}$, instead of their limit price, the payouts made to the bidders are exactly funded by the money collected from the accepted orders in the worst-case outcome. $\theta>0$ represents starting orders needed to guarantee uniqueness of these state prices in the solution. $\mu>0$ is the weight given to the starting order term.  
%The constraints in this model ensure that the market is self-funding and that the quantities granted to each participant are consistent based on the relationship of their limit price and the calculated state price of the order. Furthermore, it is valuable that the model has a unique optimum.

The significance of starting orders needs a special mention here. Without the starting orders, (\ref{CPCAM}) would be a linear program with 
%multiple optimal solutions and 
multiple dual solutions. Introducing the convex barrier term involving $\theta$ makes the dual strictly convex resulting in a unique optimal price vector. To understand its effect on the computed prices, consider the dual problem for (\ref{CPCAM}): %can be stated as follows:
\begin{equation}\label{CPCAM-dual}
\begin{array}{lll}
\displaystyle\min_{y, p} & q^Ty - \mu \sum_{i=1}^S \theta_i \log(p_i) & \nonumber\\
\mbox{s.t.} & \sum_i p_i =1 & \nonumber\\
& \sum_{i} a_{ik} p_i + y_k \ge \pi_k & \forall k \nonumber\\
& y\ge 0 & 
\end{array}
\end{equation}
Observe that if $\theta$ is normalized, the second term in the objective gives the K-L distance\footnote{
The Kullback Leibler distance (KL-distance) is a measure of the difference between two probability distributions. 
%from a "true" probability distribution, p, to a "target" probability distribution, q
The K-L distance of a distribution $p$ from $\theta$ is given by $\sum_i \theta_i \log{\frac{\theta}{p_i}}$.  
} of $\theta$ from $p$ (less a constant term $\sum \theta_i \log{\theta_i}$). 
Thus, when $\mu$ is small, the above program optimizes the first term $q^Ty$, and among all these optimal price vectors picks the one that minimizes the K-L distance of $p$ from $\theta$. 
As discussed in the introduction, the price vectors are of special interest due to their interpretation as outcome distributions. 
%From the above discussion
Thus, the starting orders enable us to choose the unique distribution $p$ that is closest (minimum K-L distance) to a prior specified through $\theta$. 

%Lange and Economides' model 
The CPCAM model shares many desirable properties with the limit order parimutuel call auction model originally developed by Lange and Economides \cite{Lange}.  
%that are shared by the Peters et al. \cite{Ye-CPCAM} formulation above.
Some of its important properties from information aggregation perspective are
1) it produces a self-funded auction, 2) it creates more liquidity by allowing multi-lateral order matching,
3) the prices generated satisfy ``price consistency constraints", that is, the market
organizer agrees to accept the orders with a limit price
greater than the calculated price of the order while
rejecting any order with a lower limit price.
The price consistency constraints ensure the traders that their orders are being duly
considered by the market organizer, and provide incentive for informed traders to trade whenever their information would change the price. Furthermore, it is valuable that the model has a unique optimum and produces a unique price vector.

Although the above model has many powerful properties, its call auction setting suffers from the drawback of a delayed decision. The traders are not sure about the acceptance of their orders until after the market is closed. Also, it is difficult to determine the optimal bidding strategy for the traders and ensure truthfulness. In a consecutive work, Peters et al. \cite{Ye-comparison} introduced a ``Sequential Convex Parimutuel Mechanism (SCPM)" which is an extension of the CPCAM model to a dynamic setting, and has additional properties of immediate decision and truthfulness in a myopic sense. The techniques discussed in this paper assume a call auction setting, but can be directly applied to this sequential extension. 

%Another property that is desirable in an automated market mechanism and is crucial for combinatorial markets is that after any trade computing the prices should be a tractable problem. 
%In the above-mentioned work \cite{Ye-CPCAM}, the complexity of price computation was studied for a simple auction market with a small outcome state space (polynomial $S$), and it cannot be directly applied for exponentially large state space in permutation betting. 

\section{Permutation Betting Mechanisms}\label{sec:mechanisms}
%In this section, we propose our `proportional betting mechanism' for permutation betting scenario. Consider a betting scenario with $n$ candidates. Traders bet on rankings of the candidates in the final outcome.
In this section, we propose new mechanisms for betting on permutations under the parimutuel call auction model described above. Consider a permutation betting scenario with $n$ candidates. Traders bet on rankings of the candidates in the final outcome. The final outcome is represented by an $n \times n$ permutation matrix, where ${ij}^{th}$ entry of the matrix is $1$ if the candidate $i$ takes position $j$ in the final outcome and $0$ otherwise. We propose betting mechanisms that restrict the admissible bet types to `set of candidate-position pairs'. Thus, the trader $k$'s bet will be specified by an $n\times n$ $(0,1)$ matrix $A_k$, with $1$ in the entries corresponding to the candidate-position pairs he is bidding on. We will refer to this matrix as the `bidding matrix' of the trader. 
If the trader's bid is accepted, he will receive some payout in the event that his bid is a ``winning bid". 

Depending on how this payout is determined, two variations of this mechanism are examined: 
%We examine two variations of this mechanism that differ in how this payout is determined 
a) Fixed Reward Betting and b) Proportional Betting. The intractability of fixed reward betting will provide motivation to examine proportional betting more closely, which is the focus of this paper.
%Depending on how this payout is determined, two kinds of mechanisms could be proposed: 

\paragraph{Fixed reward betting} In this mechanism, a trader receives a fixed payout (assume \$1 w.l.o.g.) if {\it any} entry in his bidding matrix matches with the corresponding entry in the outcome permutation matrix.
% and $0$ otherwise. 
That is, if $M$ is the outcome permutation matrix, then the payout made to trader $k$ is given by $I(A_k \bullet M>0)$. Here, the operator `$\bullet$' denotes the Frobenius inner product\footnote{
The Frobenius inner product, denoted as $A\bullet B$ in this paper, is the component-wise inner product of two matrices as though they are vectors. That is, 
$$ A\bullet B = \sum_{i,j} A_{ij}B_{ij}$$
}, and $I(\cdot)$ denotes an indicator function. The market organizer must decide which bids to accept in order to maximize the worst case profit.
Using the same notations as in the CPCAM model described in Section \ref{background} %(\ref{CPCAM}) 
for limit price, limit quantities, and accepted orders,
the problem for the market organizer in this mechanism can be formulated as follows:\\
\begin{equation}\label{0-1game}
\begin{array}{lll}
\max & \pi^Tx-r & \\
\mbox{s. t.} & r \ge \sum_{k=1}^{m} I(A_k\bullet M_\sigma>0)x_k & \forall \sigma\in {\cal S}_n\\%, 1 \le k \le m\\
%& 0 \le z_k \le 1 & 1 \le k \le m\\
 & 0 \le x \le q & 
\end{array}
\end{equation}
Here, ${\cal S}_n$ represents the set of $n$ dimensional permutations, $M_\sigma$ represents the permutation matrix corresponding to permutation $\sigma$. 
Note that this formulation encodes the problem of maximizing the worst-case profit of the organizer with no starting orders.

Above is a linear program with exponential number of constraints. We prove the following theorem regarding the complexity of solving this linear program.

\begin{theorem}\label{T1}
The optimization problem in (\ref{0-1game}) is NP-hard even for the case when there are only two non-zero entries in each bidding matrix.
\end{theorem}
\begin{proof}
 The separation problem for the linear program in (\ref{0-1game}) corresponds to finding the permutation that ``satisfies" maximum number of bidders. Here, an outcome permutation is said to ``satisfy" a bidder, if his bidding matrix has at least one coincident entry with the permutation matrix. We show that the separation problem is NP-hard using a reduction from maximum satisfiability (MAX-SAT) problem. 
%We show a reduction from the maximum satisfiability problem. 
%We present an outline of the proof here. 
In this reduction, the clauses in the MAX-SAT instance will be mapped to bidders in the bidding problem. And, the number of non-zero entries in a bidding matrix will be equal to the number of variables in the corresponding clause. 
%	Given an instance of the boolean MAX-SAT problem, we map the variables to candidates and their possible assignments to positions in the bidding problem. Also, we map each clause to a bidder; the bidding matrix is constructed by entering $1$s for the ``variable-assignment" pairs that appear in the clause. The satisfiability of clauses will correspond to the satisfiability of bidders and vice-versa. 
Since, MAX-2-SAT is NP-hard, this reduction will prove the NP-hardness even for the case when each bidding matrix is restricted to have only two non-zero entries.
See the appendix for the complete reduction.  

Using the result on equivalence of separation and optimization problem from \cite{grotschel}, the theorem follows.
\end{proof}

 This result motivates us to examine the following variation of this mechanism which makes payouts proportional to the number of winning entries in the bidding matrix. 

\paragraph{Proportional betting} In this mechanism, the trader receives a fixed payout (assume \$1 w.l.o.g.) {\it for each coincident entry} between the bidding matrix $A_k$ and the outcome permutation matrix. Thus, the payoff of a trader is given by the Frobenius inner product of his bidding matrix and the outcome permutation matrix.
%The market organizer must decide which bids to accept in order to maximize the worst case profit.
%Using the same notations as in CPCAM model (\ref{CPCAM}) for limit price, limit quantities, and accepted orders,
The problem for the market organizer in this mechanism can be formulated as follows:\\
\begin{equation}\label{basicProblem}
\begin{array}{lll}
\max & \pi^Tx-r & \\
\mbox{s. t.} & r \ge \sum_{k=1}^{m} (A_k\bullet M_\sigma)x_k & \forall \sigma\in {\cal S}_n \\
 & 0 \le x \le q & \\
%& M = \begin{array}{ll}
%	\max& {(A_1x_1+\ldots+A_mx_m)\bullet P}\nonumber\\
%	\mbox{s.t.} & \textrm{P is a $n\times n$ permutation matrix} \nonumber
%	\end{array} &
\end{array}
\end{equation}

%\shcomment{TODO: Compare with subset betting}
The above linear program involves exponential number of constraints. However, the separation problem for this program is polynomial-time solvable, since it corresponds to finding the maximum weight matching in a complete bipartite graph, where weights of the edges are given by elements of the matrix $(\sum_k A_kx_k)$. %The maximum matching problem can be solved in polynomial time. 
Thus, the ellipsoid method with this separating oracle would give a polynomial-time algorithm for solving this problem. This approach is similar to the algorithm proposed in \cite{nikolova} for {\it Subset Betting}. Indeed, for the case of subset betting \cite{nikolova}, the two mechanisms proposed here are equivalent.
%making proportional payouts is equivalent to making fixed payouts. 
This is because subset betting can be equivalently formulated under our framework, as a mechanism that allows non-zero entries only on a single row or column of the bidding matrix $A_k$. Hence, the number of entries that are coincident with the outcome permutation matrix can be either $0$ or $1$, resulting in $I(A_k\bullet M_\sigma >0) = A_k \bullet M_\sigma$, for all permutations $\sigma$. Thus, subset betting forms a special case of the proportional betting mechanism proposed here, and all the techniques derived in the sequel for proportional betting will directly apply to it.

%Next, we study the pricing problem for the proportional betting mechanism. In the process, we will also give a more compact formulation for the organizer's problem (\ref{basicProblem}) involving only $n^2$ constraints. Hence, in practice, the new formulation will be much faster to solve using interior point methods. 

\section{Pricing in Proportional Betting}\label{sec:propoPricing}
%The dual variables for the $n!$ constraints in (\ref{basicProblem}) correspond to the prices for the $n!$ permutations. Also, exists a basic dual solution %to the above problem with only polynomial number of non-zeros, which can be constructed using a polynomial time algorithm (Lemma $6.5.15$ in \cite{grotschel}). Thus, one could obtain a price vector of polynomial size for the above market. However, the price vector thus obtained is not guaranteed to be unique.  To ensure unique prices, we could use starting orders as discussed in the CPCAM model in Section \ref{background}. This would introduce one starting order and slack variable corresponding to each permutation. But, the problem would then have exponential number of variables as well as constraints. Also, since the resulting problem will not be linear (due to the barrier term $\sum_\sigma \theta_\sigma \log({s_\sigma)}$ in the objective), the dual price vector is not guaranteed to have polynomial support.
%To overcome these difficulties, we reformulate the above problem into an equivalent linear program with only $n^2$ constraints. After introducing $n^2$ starting orders, the dual variables for these constraints will give us $n^2$ ``unique marginal prices".
In this section, we reformulate the market organizer's problem for Proportional Betting 
%in (\ref{basicProblem})
into a compact linear program involving only $O(n^2+m)$ constraints. 
Not only the new formulation is faster to solve in practice (using interior point methods), but also it will generate a compact dual price vector of size $n^2$. These `marginal prices' will be sufficient to price the bets in Proportional Betting, and are shown to satisfy some useful properties. The reformulation will also allow introducing $n^2$ starting orders in order to obtain unique prices.

%Observe that the first constraint in (\ref{basicProblem}) implicitly involves the problem of finding the permutation (or matching) with worst payoff $r$ for a given decision $x$. 
Observe that the first constraint in (\ref{basicProblem}) implicitly sets $r$ as the worst case payoff over all possible permutations (or matchings).
% with worst payoff $r$ for a given decision $x$. 
%Since the matching polytope is integral \cite{grotschel}, this problem can be stated as the following linear program that computes maximum weight matching:\\
Since the matching polytope is integral \cite{grotschel}, $r$ can be equivalently set as the result of following linear program that computes maximum weight matching:\\
\begin{equation}
\begin{array}{llll}
r = & \displaystyle\max_M & (\sum_{k=1}^{m} x_kA_k) \bullet M &\nonumber\\
 & \mbox{s.t.} & M^Te=e&\\
 && Me = e&\\
 && M_{ij} \ge 0& 1\le i,j \le n 
\end{array}
\end{equation}
Here $e$ denotes the vector of all $1$s (column vector). 
%\shcomment{Should we use something other than "e", since later on e is used a lot as exponent. I have seen boldfaced "1" being used as vector of 1s in some papers.}
Taking dual, equivalently,
\begin{equation}
\begin{array}{llll}
r = & \displaystyle\min_{v,w} & e^Tv+e^Tw &  \nonumber\\
& \textrm{s.t.} & v_i+w_j\geq \sum_{k=1}^{m} (x_kA_k)_{ij} &\forall i,j
\end{array}
\end{equation}
Here, $(x_kA_k)_{ij}$ denotes the $ij^{th}$ element of the matrix $(x_kA_k)$.
The market organizer's problem in (\ref{basicProblem}) can now be formulated as:
\begin{equation}\label{primal-market}
\begin{array}{lll}
\displaystyle\max_{x,v,w} & \pi^Tx-e^Tv-e^Tw &\\
\textrm{s.t.} & v_i+w_j\geq \sum_{k=1}^m (x_kA_k)_{ij} &\forall i,j\\
& 0\leq x\leq q &\\
\end{array}
\end{equation}
\comment{
and its dual:
\begin{equation}\label{dual-market}
\begin{array}{lll}
\displaystyle\min_{y,Q} & q^{T}y &\\
\textrm{s.t.} & Qe=e &\\
& Q^{T}e=e &\\
& A_k\bullet Q+y_k\geq \pi_k & \forall k\\
& y\geq 0 &\\
%& Q\geq 0 &\\
\end{array}
\end{equation}

%We stated in the previous chapter 
%to this problem satisfy the required price consistency constraints and 
We show that the dual variable $Q$ 
can be interpreted as the price %or probability 
of each entry in the bidding matrix. However, we can not assure that this dual
solution is unique. 
In the sequel, we obtain unique price matrix using the idea of
seeding each state with a very small order. This idea is similar
to the one proposed in \cite{Ye-CPCAM}.
}
Observe that this problem involves only $n^2+2m$ constraints. As we show later, the $n^2$ dual variables for the first set of constraints can be well interpreted as marginal prices. However, the dual solutions for this problem are not guaranteed to be unique.  To ensure uniqueness, we can use starting orders as discussed for the CPCAM model in Section \ref{background}. After introducing one starting order $\theta_{ij}>0$ for each candidate-position pair, and slack variables $s_{ij}$ for each of 
the $n^2$ constraints, we get the following problem:
\begin{equation}\label{primalb}
\begin{array}{lll}
\displaystyle\max_{{x,v,w,s}} & \pi^Tx-e^Tv-e^Tw+\sum_{i,j}\theta_{ij}\log(s_{ij})& \\
\textrm{s.t.} & v_i+w_j-s_{ij}= \sum_{k=1}^m (x_kA_k)_{ij}& \hspace{-0.3in}\forall i,j\\
& s_{ij}\geq0 &\hspace{-0.3in}\forall i,j\\
& 0\leq x\leq q &\\
\end{array}
\end{equation}
\noindent and its dual:
\begin{equation}\label{dualb}
\begin{array}{lll}
\displaystyle\min_{{y,Q}} & q^{T}y-\sum_{i,j}\theta_{ij}\log(Q_{ij}) &\\
\textrm{s.t.} & Qe=e &\\
& Q^{T}e=e & \\
& A_k\bullet Q+y_k\geq \pi_k & \forall k\\
& y\geq 0 &\\
\end{array}
\end{equation}

%Here, the logarithmic term can be viewed as a barrier function, which is used widely in solving optimization problems. 
Next, we
will show that model (\ref{primalb}) and (\ref{dualb}) possess
many desirable characteristics.
\begin{lemma} 
Model (\ref{primalb}) and (\ref{dualb}) are convex programs. And if
$\theta_{ij}>0, \forall i,j$, the solution to (\ref{dualb}) is
unique in $Q$.
\end{lemma} 

\begin{proof}
Since logarithmic function is concave and the constraints are
linear, we can easily verify that (\ref{primalb}) and
(\ref{dualb}) are convex programs. Also, according to our
assumption on $\theta$, the objective function in (\ref{dualb}) is strictly
convex in $Q$. Thus, the optimal solution of (\ref{dualb}) must be
unique in $Q$.
\end{proof}

Therefore, we know that this program can be solved up to any given
accuracy in polynomial time using convex programming methods and
produces unique dual solution $Q$.

We show that the dual matrix $Q$ generated from (\ref{dualb}) is
well interpreted as a ``parimutuel price". That is, $Q\ge 0$; and, if we 
charge each trader $k$ a price of $A_k \bullet Q$ instead of their limit price ($\pi_k$), then the optimal decision remains unchanged and 
the total premium paid by the accepted orders will be equal to the total payout made in the worst case. 
Further, we will show that $Q$ satisfies the following extended definition of  ``price
consistency condition'' introduced in \cite{Lange}.

\begin{definition} \label{PCC}
The price matrix $Q$ satisfies the price consistency
constraints if and only if for all $j$:\vspace{-0.2in}
\begin{equation}\label{pcc}
\begin{array}{lll}
x_j=0 &\Rightarrow & Q\bullet A_j=c_j\geq\pi_j \nonumber\\
0<x_j<q_j &\Rightarrow & Q\bullet A_j = c_j=\pi_j\\
x_j=q_j & \Rightarrow & Q\bullet A_j = c_j\leq\pi_j\\
\end{array}
\end{equation}
That is, a trader's bid is accepted only if his limit price is greater than the calculated price for the order. 
\end{definition}

%\begin{lemma}
%The unique solution $Q$ to (\ref{dualb}) gives a `parimutuel price' matrix, and satisfies the `price-consistency constraints' defined in Definition \ref{PCC}.
%\end{lemma}
%\begin{proof}
To see this,  we
construct the Lagrangian function for program (\ref{primalb}):
$$
\begin{array}{l}
L(x, Q, s, v, w, y)  \nonumber\\
 \hspace{0.15in}       =  \hspace{0.1in}\pi^Tx-e^Tv-e^Tw+\sum_{i,j}\theta_{ij}\log{s_{ij}}\nonumber\\
 \hspace{0.15in}       -  \hspace{0.1in}\sum_{i,j}Q_{ij}(s_{ij}+\sum_k (x_kA_k)_{ij}-v_i-w_j)\nonumber\\
 \hspace{0.15in}       +  \hspace{0.1in}\sum_{i}y_i(q_i-x_i)\nonumber
\end{array}
$$
\noindent Now, we can derive the KKT conditions:
\begin{equation}%\label{KKT}
\begin{array}{ll}
\pi_k-Q\bullet A_k-y_k\leq0 &  1\leq k \leq m \nonumber\\
x_k\cdot(\pi_k-Q\bullet A_k-y_k)=0 &  1\leq k \leq m \\
Qe=e&\\
Q^Te=e&\\
\frac{\theta_{ij}}{s_{ij}}-Q_{ij}\le 0 & 1\leq i,j\leq n\\
s_{ij}\cdot(\frac{\theta_{ij}}{s_{ij}}-Q_{ij})= 0 & 1\leq i,j\leq
n\\
y_k\cdot(x_k-q_k)=0&  1\leq k \leq m\\
y\geq 0 &\\
\end{array}
\end{equation}

Since $s_{ij} > 0$ for any optimal solution, the above conditions imply that $Q_{ij} = \frac{\theta_{ij}}{s_{ij}}$, or
$s_{ij} = \frac{\theta_{ij}}{Q_{ij}}$ for all $ij$. Since, $\theta_{ij} >0$, this implies $Q_{ij} > 0$, for all $ij$. Also, the first constraint in the primal problem (\ref{primalb}) now
gives: $v_i + w_j = \sum_{k} (x_kA_k)_{ij} + \frac{\theta_{ij}}{Q_{ij}}$. Multiplying with $Q_{ij}$, and summing over all $i,j$:\vspace{-0.1in}
\begin{center}
$r=e^Tv+e^Tw = \sum_{k} x_k(A_k\bullet Q) + \sum_{ij} \theta_{ij}$
\end{center}
Since, $r$ gives the worst case payoff, charging the bidders according to price matrix $Q$ 
results in a parimutuel market (except for the amount invested in the starting orders, an issue that we address later).  
%\begin{proposition} 
%The dual matrix generated by model (\ref{dualb}) must satisfy
%the price consistency constraints.
%\end{proposition}
%\begin{proof}
%The proof is referred to the proof in \cite{Ye-CPCAM}.
Also, if we replace $\pi_k$ with $A_k \bullet Q$ in the above KKT conditions and set $y_k=0$, the solution $x, s, Q$ will still satisfy all the KKT conditions. Thus, the optimal solution remains unchanged. 
Further, observe that the first two conditions along with the
penultimate one are exactly the price consistency constraints. Hence, $Q$
%the dual matrix generated by model (\ref{dualb}) 
must satisfy the price consistency constraints.
%\end{proof}

%\shcomment{TODO: Theorem that optimal remains same if parimutuel price is charged}
%One may notice that the above model involves new parameters
%$\theta_{ij}$. In fact, these 
In the above model, market organizer needs to seed the market with the starting orders $\theta_{ij}$ 
in order to ensure uniqueness of the optimum state price matrix. 
The market organizer could actually lose this
seed money in some outcomes. In practice, we can set the
starting orders to be very small so that this is not an issue.
On the other hand, it is natural to ask whether the starting
orders can be removed altogether from the model to make the market absolutely parimutuel. The following lemma shows that
this is indeed possible.
\begin{lemma}
For any given starting orders $\theta$, as we reduce $\theta$ uniformly to $0$,
the price matrix converges to a unique limit $Q$,
%there is a unique limit matrix $Q$ as we reduce $\theta$ to zero with each element reduced by the same proportion, 
and this limit is an optimal dual
price for the model without the starting orders given in (\ref{primal-market}).
\end{lemma}
\begin{proof}
The proof of this lemma follows directly from the discussion in Section 3.1 of \cite{Ye-CPCAM}.
\end{proof}
Moreover, as discussed in \cite{Ye-CPCAM}, such a limit $Q$ 
can be computed efficiently using the path-following algorithm developed in \cite{Ye-path}.

%From the above discussion, we can see that the price matrix $Q$ of our model shares many desirable properties with the price vector in CPCAM model \cite{Ye-CPCAM}. These properties include the uniqueness and consistency of price, and the limiting behavior. 
%The price matrix $Q$ is a marginal price vector in the sense that it allocates a price to each candidate-position pair. From the perspective of organizing the market, it is sufficient to compute the marginal price matrix $Q$ and charge each bidder $k$ a total price of $A_k \bullet Q$. 
To summarize, we have shown that:

\begin{theorem}
One can compute in polynomial-time, an $n\times n$ marginal price matrix $Q$ 
%can be computed 
which is sufficient to price the bets in the Proportional Betting mechanism. Further, the price matrix is unique, parimutuel, and satisfies the desired price-consistency constraints.
%, and converges to optimal as starting orders are driven to $0$.
\end{theorem}
%Also, since $Q$ is a doubly stochastic matrix, each row and column of $Q$ gives a probability distribution. The distribution given by row $i$ of $Q$ could be interpreted as the marginal distribution over ranks $1$ to $n$ for candidate $i$. Similarly, column $j$ gives marginal distribution over candidates for rank $j$. 

% However, for the purpose of information collection, one would like to compute the complete price vector that allocates a price to each of the $n!$ outcomes. This price vector would be interpreted as a joint probability distribution over the entire state space. In the next section, we discuss methods for computing this complete price vector from the marginal prices given by $Q$.

%one important feature of the contingent claim markets is that it can collect believes from the bidders and provide predictive information . In the combinatorial auction case, the price can be interpreted as the probability at which each outcome might happen. In our case, we are interested in finding the probability of a certain outcome, e.g., a certain rank of the race, is realized. The price matrix $Q$ computed from (\ref{dual-market}) just gives us the marginal probability of each state, but not the probability over the permutations. In next chapter, we will show how to extract the probability of each permutation from the price matrix $Q$ under certain criteria.

%%%%%%%%%%%%%%%%%%%%%%%%%%%%%%%%%%%%%%%%%%%%%%%%%%%%%%%%%%%%%%%%%%%%%%%%%%%%%%%%%%%%%%%%%%%%%%5

\section{Pricing the Outcome Permutations}% in proportional betting}
%joint distribution over permutations}
There is analytical as well as empirical evidence that prediction market prices provide useful estimates of
average beliefs about the probability that an event occurs \cite{adams06, manski06, ottaviani06, wolfers06}.
Therefore, prices associated with contracts are typically treated as predictions of the probability of future events. 
The marginal price matrix $Q$ derived in the previous section associates a price to each candidate-position pair. Also, observe %from the constraints in problem (\ref{dualb}) 
that $Q$ is a doubly-stochastic matrix (refer to the constraints in dual problem (\ref{dualb})).
%Also, since $Q$ is a doubly stochastic matrix, each row and column of $Q$ gives a probability distribution. 
%Thus, the distribution given by row $i$ of $Q$ could be treated as the marginal distribution over positions $1$ to $n$ for candidate $i$. Similarly, column $j$ gives marginal distribution over candidates for rank $j$. 
Thus, the distributions given by a row (column) of $Q$ could be interpreted as marginal distribution over positions for a given candidate (candidates for a given position). 
One would like to compute the complete price vector that assigns a price to each of the $n!$ outcome permutations. This price vector would provide information regarding the joint probability distribution over the entire outcome space. In this section, we discuss methods for computing this complete price vector from the marginal prices given by $Q$.
%In this section, our goal is to compute a joint probability distribution over all $n$ dimensional permutations whose marginal distributions are given by $Q$. 

%As discussed above, the price matrix $Q$ obtained from our formulation can be treated as marginal distribution over the outcome space. 

Let $p_\sigma$ denote the price %\footnote{In the following discussion, we will use the terms price and probabilities interchangeably.} 
for permutation $\sigma$.
Then, the constraints on the price vector $p$ are represented as:\vspace{-0.1in}
\begin{equation}\label{feasible-price}
\begin{array}{rcll}
%\min & \sum_{\sigma\in\mathbb{S}_n}{p_\sigma \log{p_\sigma}} \\
%\text{s.t.} & 
\sum_{\sigma \in {\cal S}_n}{p_\sigma M_\sigma} & = & Q & \\
p_\sigma & \geq & 0 & \forall \sigma \in {\cal S}_n
\end{array}
\end{equation}
%Here ${\cal S}_n$ is the set of permutations of size $n$, $M_\sigma$ is an $n\times n$ matrix representing the permutation $\sigma$. That is, the $ij^{th}$ entry of $M_\sigma$, $M_{\sigma,ij}=1$ if $j=\sigma(i)$, and $0$ otherwise. 
Note that the above constraints implicitly impose the constraint
$\sum_\sigma p_\sigma =1$. Thus, $\{p_\sigma\}$ is a valid distribution. Also, it is easy to establish that if $Q$ is an optimal marginal price matrix, then any such $\{p_\sigma\}$ is an optimal joint price vector over permutations.  That is,
\begin{lemma}\label{L4}
If $Q$ is an optimal dual solution for (\ref{primal-market}), then any price vector $\{p_\sigma\}$ that satisfies the constraints in (\ref{feasible-price}) is an optimal dual solution for (\ref{basicProblem}).
\end{lemma}
\begin{proof} The result follows directly from the structure of the two dual problems. See appendix for a detailed proof. \end{proof}

Finding a feasible solution under these constraints is equivalent to finding a decomposition of doubly-stochastic matrix
$Q$ into a convex combination of $n\times n$ permutation matrices. There are multiple such decompositions possible. For example, one such 
solution can be obtained using Birkhoff-von Neumann decomposition \cite{birkhoff, dulmage}. Next, we propose a criterion to choose a meaningful distribution $p$ from the set of distributions satisfying constraints in (\ref{feasible-price}).

\subsection{Maximum entropy criterion} Intuitively, we would like to use all the information about the marginal distributions that we have, but avoid including any information that we do not have. 
%Therefore one should be as uncommitted as possible about missing information.
%that satisfies these constraints but is otherwise as random as possible.
%That is, it makes no assumption other than these constraints. 
This intuition is captured by the `Principle of Maximum Entropy'. 
It states that the least biased distribution that encodes certain given information is that which maximizes the information entropy.
%which gives the `closest distribution to a uniform prior'.

Therefore, we consider the problem of finding the maximum entropy distribution over the space of $n$ dimensional permutations, satisfying the above 
constraints on the marginal distributions. The problem can be represented as follows: 
\begin{equation}\label{max-entropy}
\begin{array}{ll}
\min & \sum_{\sigma\in {\cal S}_n}{p_\sigma \log{p_\sigma}} \\
\text{s.t.} & \sum_{\sigma \in {\cal S}_n}{p_\sigma M_\sigma}=Q \\
& p_\sigma\geq0
\end{array}
\end{equation}
The maximum entropy distribution obtained from above has many nice properties. 
Firstly, as we show next, the distribution has a concise representation in terms of only $n^2$ parameters. This property is 
crucial for combinatorial betting due to the exponential state space over which the distribution is defined.
%Further, we show that this distribution is also the ``maximum-likelihood distribution" under suitable interpretation of $Q$. 
Let $Y \in R^{n\times n}$ be the Lagrangian dual variable corresponding to the marginal distribution constraints in (\ref{max-entropy}), and $s_\sigma$ be the dual variables corresponding to non-negativity constraints on $p_\sigma$. Then, the KKT conditions for (\ref{max-entropy}) are given by:
%Then, from the KKT conditions, observe that following condition is satisfied at optimality:
\begin{equation}\label{KKT-max-entropy}
\begin{array}{rcll}
\log({p_\sigma}) + 1 - s_\sigma & = & Y\bullet M_\sigma & \\
\sum_\sigma p_\sigma M_\sigma & = & Q \\
s_\sigma, p_\sigma & \ge & 0 & \forall \sigma\\
p_\sigma s_\sigma & = & 0 & \forall \sigma
\end{array}
\end{equation}
Assuming $p_\sigma>0$ for all $\sigma$, this gives $p_\sigma = e^{Y\bullet M_\sigma-1}$.
%\footnote{This correspondence between maximum entropy distributions and exponential family of distributions is well known and has been observed earlier in \cite{asadpour} and in various contexts in machine learning literature.}. 
%Thus, for any permutation $\sigma$, its price $p_\sigma$ is given by the product $\Pi_{ij: j=\sigma(i)} e^{Y_{ij}}$. 
Thus, the distribution is completely specified by the $n^2$
parameters given by $Y$. Once $Y$ is known, it is possible to perform operations like computing the probability for a given set of outcomes, or sampling the highly probable outcomes. 

Further, we show that the dual solution $Y$ is a maximum likelihood estimator of distribution parameters under suitable interpretation of $Q$.
\paragraph{Maximum likelihood interpretation}
For a fixed set of data and an {\it assumed} underlying probability model, maximum likelihood estimation method picks the values of the model parameters that make the data ``more likely" than any other values of the parameters would make them. 
%Consider the exponential family $D_\eta$ of probability distributions parameterized by an unknown parameter $\eta$, associated with either a known probability density function (continuous distribution) or a known probability mass function (discrete distribution), denoted as f?. We draw a sample  of n values from this distribution, and then using f? we compute the (multivariate) probability density associated with our observed data, 
%As a function of ? with x1, ..., xn fixed, this is the likelihood function
%The method of maximum likelihood estimates ? by finding the value of ? that maximizes . This is the maximum likelihood estimator (MLE) of ?:
Let us assume in our model that the traders' beliefs about the outcome come from an exponential family of distributions $D_\eta$, with probability density function of the form $f_\eta \propto e^{\eta \bullet M_\sigma}$ for some parameter $\eta \in R^{n\times n}$. Suppose $Q$ gives a summary statistics of 
%we could observe 
$s$ sample observations $\{M^1, M^2, \dots, M^s\}$ from the traders' beliefs, 
%and $Q$ gives the summary statistics of these observations, 
i.e., $Q=\frac{1}{s} \sum_{k} M^k$. This assumption is inline with the interpretation of the prices in prediction markets as mean belief of the traders. 

Then, the maximum likelihood estimator $\widehat{\eta}$ of $\eta$ is the value that maximizes the likelihood of these observations, that is:
\begin{equation}
\begin{array}{rcl}
\widehat{\eta} & = & \mbox{arg} \max_\eta \log f_\eta(M^1, M^2, \ldots, M^s) \nonumber\\
	& = & \mbox{arg} \max_\eta \log (\Pi_k \frac{e^{\eta \bullet M^k}}{\sum_\sigma e^{\eta \bullet M_\sigma} }) \\
\end{array}
\end{equation}
%The maximum likelihood estimator is given by 
The optimality conditions for the above unconstrained convex program are:
\begin{equation}
\begin{array}{rcl}
\frac{1}{Z}\sum_\sigma e^{\eta \bullet M_\sigma} M_\sigma & = & \frac{1}{s}\sum_k M^k\nonumber
\end{array}
\end{equation}
where $Z$ is the normalizing constant, $Z=\sum_\sigma e^{\eta \bullet M_\sigma}$.
%As mentioned earlier, dual price vector $Q$ is interpreted as the sample mean of the traders' beliefs. Hence, substituting the sample mean $\sum_k M^k/s$ by $Q$, 
Since $\frac{1}{s}\sum_k M^k = Q$, observe from the KKT conditions for the maximum entropy model given in (\ref{KKT-max-entropy}) that $\eta=Y$ satisfies the above optimality conditions. Hence, the parameter $Y$ computed from the maximum entropy model is also the maximum likelihood estimator for the model parameters $\eta$. 
%It is well known in the machine learning literature that the maximum likelihood distribution in the exponential family is the distribution with expected value same as the sample mean $\frac{1}{s} \sum_{i=1}^{s} M^i$. As we discussed earlier, the price matrix $Q$ can be interpreted as the mean belief of traders about the outcome. Thus, interpreting $Q$ as the sample mean,  observe that the maximum entropy distribution derived above is a distribution in the exponential family ($p_\sigma \propto e^{Y\bullet M_\sigma}$), and satisfies the constraints that expected value $\sum x_\sigma M_\sigma$ is equal to the sample mean $Q$. Thus, we can conclude that it is also the maximum likelihood distribution.

%\shcomment{TODO: add reference, add remark about non-uniform prior} 
%Furthermore, this calculation shows that the maximum entropy distribution belongs to an ``exponential family". It is well known in the machine learning literature that ...
%\paragraph{Maximum-likelihood interpretation}

%%%%%%%%%%%%%%%%%%%%%%%%%%%%%%%%%%%%%%%%%%%%%%%%%%%%%%%%%%%%%%%%%%%%%%%%
\subsection{Complexity of the Maximum Entropy Model}
In this section, we analyze the complexity of solving the maximum entropy model in (\ref{max-entropy}). As shown in the previous section, the solution to this model is given by the parametric distribution $p_\sigma=e^{Y\bullet M_\sigma-1}$. The parameters $Y$ are the dual variables given by the optimal solution to the following dual problem of (\ref{max-entropy})
\begin{equation}\label{dual-max-entropy}
\begin{array}{ll}
\displaystyle\max_Y & Q\bullet Y- \sum_\sigma e^{Y\bullet M_\sigma-1}\\
\end{array}
\end{equation}
We prove the following result regarding the complexity of computing the parameters $Y$:
%\shcomment{TODO: Refine this proof}
\begin{theorem}\label{Hardness-entropy}
It is \#P-hard to compute the parameters of the maximum entropy distribution 
%given by $(\ref{max-entropy})$.
$\{p_\sigma\}$ over $n$ dimensional permutations $\sigma \in {\cal S}_n$, that has a given marginal distribution.  
\end{theorem}
\begin{proof}
%We focus on proving the hardness of computing the value of dual variable $Y$. Since, the optimal primal and dual solutions
%are related as $p_\sigma=e^{Y\bullet M_\sigma-1}$, the result will imply the above lemma. 
%The optimal $Y$ is given by the solution of the dual problem:
%\begin{equation}\label{dual-max-entropy}
%\begin{array}{ll}
%\displaystyle\max_Y & Q\bullet Y- \sum_\sigma e^{Y\bullet M_\sigma-1}\\
%\end{array}
%\end{equation}
%To prove the hardness of this problem, 
We make a reduction from the following problem:\\
{\bf Permanent of a $(0,1)$ matrix} The permanent of an $n\times n$ matrix $B$ is defined as $\mbox{perm}(B) = \sum_{\sigma \in {\cal S}_n} \Pi_{i=1}^n B_{i,\sigma(i)}$. Computing permanent of a $(0,1)$ matrix is \#P-hard \cite{valiant-perm}.
%The optimal solution to the above problem is given by solution the following set of equations (obtained by setting derivative to 0):

We use the observation that  $\sum_\sigma e^{Y\bullet M_\sigma} = \mbox{perm}(e^Y)$, where the notation $e^Y$ is used to mean component-wise exponentiation: $(e^Y)_{ij} = e^{Y_{ij}}$. For complete proof, see the appendix.
\end{proof}

Interestingly, there exists an FPTAS based on MCMC methods for computing the permanent of any non-negative matrix \cite{sinclair}. Next, we derive a polynomial-time algorithm for approximately computing the parameter $Y$  
that uses this FPTAS along with the ellipsoid method for optimization. 
%solving the above dual problem. 
%The algorithm uses the FPTAS for computing the permanent along with the ellipsoid method for optimization. 
%We also discuss how to obtain an approximate solution to our original problem in (\ref{max-entropy}) from this approximate dual solution.

%%%%%%%%%%%%%%%%%%%%%%%%%%%%%%%%%%%%%%%%%%%%%%%%%%%%%%%%%%%%%%%%%%%%%%%%

\subsection{An Approximation Algorithm}
In this section, we develop an approximation algorithm to compute the parameters $Y$.
We first relax the formulation in (\ref{max-entropy}) to get an equivalent problem that will lead to a better bounded dual.
% for the joint distribution as discussed above. 
%In this section, we develop algorithms to solve the maximum entropy problem in (\ref{max-entropy}).

%Since $\sum{p_\sigma}=1$, we have that the optimal solution of the above problem is the same as that of:
%\begin{equation}\label{2}
%\begin{array}{ll}
%\min & \sum{p_\sigma(\log{p_\sigma}-1)} \\
%\text{s.t.} & \sum{p_\sigma M_\sigma}=Q \\
%& p_\sigma\geq0
%\end{array}
%\end{equation}
%We claim that the optimal solution remains the same even if we relax the equality constraint to inequality, i.e.,
Consider the problem below:
\begin{equation}\label{primal-entropy}
\begin{array}{ll}
\min & \sum{p_\sigma(\log{p_\sigma}-1)} \\
\text{s.t.} & \sum{p_\sigma M_\sigma}\leq Q \\
& p_\sigma\geq0
\end{array}
\end{equation}
We prove the following lemma:
\begin{lemma}\label{C1}
The problem in (\ref{primal-entropy}) has the same optimal solution as (\ref{max-entropy}).
\end{lemma}
\begin{proof}
See the appendix.
\end{proof}
%\paragraph{Lagrangian Dual}
The Lagrangian dual of this problem %formulated in (\ref{primal-entropy}) 
is given by:
\begin{equation}\label{dual-tmp}
\begin{array}{ll}
\max & Q\bullet Y- \sum_\sigma e^{Y\bullet M_\sigma}\\
\text{s.t.} & Y\leq 0\\
\end{array}
\end{equation}
Note that $Y$ is bounded from above. 
Next, we establish lower bounds on the variable $Y$. These bounds will be useful in proving polynomial 
running time for ellipsoid method.
\begin{lemma}\label{C2}
The optimal value $OPT$ and the optimal solution $Y$ to (\ref{dual-tmp}) satisfy the following bounds:
\begin{equation}
\begin{array}{llllll}
0 & \ge & OPT & \ge & -n\log{n}-1 & \nonumber\\
0 & \ge & Y_{ij} & \ge & -n\log{n}/q_{min} & \forall i,j
\end{array}
\end{equation}
%\begin{equation}\label{boundOPT} 
%0 \ge OPT \ge -n\log{n}-1
%\end{equation}
%\begin{equation}\label{boundY}
%0 \ge Y_{ij} \ge -n\log{n}/q_{min}
%\end{equation}
Here, $q_{min} = \min\{Q_{ij}\}$. 
%|Q_{ij}\ne 0\}$.
\end{lemma}
\begin{proof}
See the appendix.
\end{proof}
\paragraph{\it Remark:} %Note that since $Q_{ij}=\theta_{ij}/s_{ij}>0$ for all $ij$ (observed in Section~\ref{sec:propoPricing}), $q_{min}>0$ if $\theta_{ij}>0$. If $\theta_{ij}=0$, so 
Note that if $Q_{ij}=0$ for any $i,j$, in a pre-processing step we could set the corresponding $Y_{ij}$ to $-\infty$ and remove it from the problem.  So, w.l.o.g. we can assume $q_{min}>0$. However, some $Q_{ij}$ could be very small, making the above bounds very large. 
One way to handle this would be to set very small $Q_{ij}$s (say less than $\delta$ for some small $\delta>0$) 
to $0$, and remove the corresponding $Y_{ij}$s from the problem. This would introduce a small additive approximation of $\delta$ 
in the constraints of the problem, but ensure that $q_{min} > \delta$.\\

From KKT conditions for the above problem, we obtain that $p_\sigma=e^{Y\bullet M_\sigma}$ at optimality. 
Substituting $p_\sigma$ into the primal constraints $\sum {p_\sigma}=1$ and $\sum {p_\sigma}M_\sigma \le Q$, we
can obtain the following equivalent dual problem with additional constraints: 
%Using the additional constraints from the primal, we can get the following 
%alternate dual problem,
%\begin{equation}\label{4}
%\begin{array}{ll}
%\max & Q\bullet Y-1\\
%\sum_\sigma e^{Y\bullet M_\sigma}M_\sigma \le Q
%\text{s.t.} & Y\leq 0
%\end{array}
%\end{equation}
%\paragraph{Additional constraints} 
%Above were obtained by substituting $p_\sigma$ in the constraints of the primal problem. 
%\paragraph{Alternate dual}
%which will have the same optimal value as the primal problem:
%Thus we can formulate the following alternate dual problem,
\begin{equation}\label{dual-entropy}
\begin{array}{llll}
&\max & Q\bullet Y -1 &\\
%\text{s.t.} & \sum{e^{Y\bullet M_\sigma}}\leq 1 &\\
& \text{s.t.} & \sum{e^{Y\bullet M_\sigma} M_\sigma}  \le Q &\\
&& Y_{ij}\ge (-n\log{n})/q_{min}  & \forall i,j\\
&& Y_{ij} \le 0 & \forall i,j\\
\end{array}
\end{equation}

The problem %in  (\ref{alternate-dual}) 
can be equivalently formulated as that of finding a feasible point in the convex body ${\textbf K}$ defined as:
\begin{equation}
\begin{array}{lll}
\textbf{K:} & Q\bullet Y -1\ge t & \nonumber\\
& \sum{e^{Y\bullet M_\sigma}}M_\sigma \leq Q & \\
& Y_{ij} \ge (-n\log{n})/q_{min} & \forall i,j\\
& Y_{ij} \leq 0 & \forall i,j
\end{array}
\end{equation}
Here, $t$ is a fixed parameter. An optimal solution to (\ref{dual-entropy}) can be found by binary search on $t \in [-n\log{n}-1,0]$. We define an approximate set $\textbf{K}_\epsilon$ by modifying the RHS of second constraint in $\textbf K$ defined above to $Q(1+\epsilon)$. Here, $\epsilon$ is a fixed parameter.

Next, we show that the ellipsoid method can be used to generate $(1+\epsilon)$-approximate solution $Y$. %to (\ref{dual-entropy}). 
We will make use of the following lemma that bounds the gradient of the convex function $f(Y) =\sum{e^{Y\bullet M_\sigma}M_\sigma}$ appearing in the constraints of the problem.

%First we analyze some properties of the convex function $f(Y) =\sum{e^{Y\bullet M_\sigma}M_\sigma}$ appearing in the constraints of the problem.
\begin{lemma}\label{C3}
For any $ij$, the gradient of the function $g(Y)=f_{ij}(Y)=\sum_\sigma e^{Y\bullet M_\sigma}(M_\sigma)_{ij}$ satisfies the following bounds:
\begin{eqnarray}
%\label{boundGrad}
||\nabla g(Y)||_2 \le & ng(Y) & \le  \sqrt{n} ||\nabla g(Y)||_2 \nonumber
\end{eqnarray}
\end{lemma} 
\begin{proof}
%Thus, an FPTAS for computing $f(Y)$ can be obtained by using the FPTAS for computing permanent of a non-negative matrix \cite{sinclair}.

%Now, we will derive bounds on the gradient of the function $f(Y)$. 
See the appendix.
\end{proof}

%----------------------------------------------------------
Now, we can obtain an approximate separating oracle for the ellipsoid method.
\begin{lemma}\label{SEP}
Given any $Y \in R^{n\times n}$, and any parameter $\epsilon > 0$, there exists an algorithm with running time polynomial in $n$, $1/\epsilon$ and $1/q_{min}$ that does one of the following:
\begin{itemize*} 
\item asserts that $Y \in \textbf{K}_{\epsilon}$
\item or, finds $C\in R^{n\times n}$ such that $C\bullet X \le C\bullet Y$ for every $X \in \textbf{K}$.
\end{itemize*}
\end{lemma}
\paragraph{Algorithm}
\begin{enumerate*}
\item If $Y$ violates any constraints other than the constraint on $f(Y)$, report $Y\notin K$. The violated inequality gives the separating hyperplane.
\item Otherwise, compute a $(1\pm\delta)$-approximation $\widehat{f}(Y)$ of $f(Y)$, where $\delta=\min\{\frac{\epsilon}{12},1\}$.
	\begin{enumerate*}
	\item If $\widehat{f}(Y) \le (1+3\delta)Q$, then report $Y \in K_\epsilon$.
	\item Otherwise, say $ij^{th}$ constraint is violated. Compute a $(1\pm\gamma)$-approximation of the gradient of the function $g(Y)=f_{ij}(Y)$, where $\gamma=\delta q_{min}/2n^4$. The approximate gradient $C=\widehat{\nabla}{g}(Y)$ gives the desired separating hyperplane.
	\end{enumerate*}
\end{enumerate*}
%\paragraph{Running time}
%As noted earlier $f(Y)$ denotes the permanent of matrix $e^Y$. Hence, the $\widehat{f}(Y)$ in Step 2 can be computed using the FPTAS for computing the permanent of matrix $e^Y$ \cite{sinclair}. Similarly, the gradient of $f(Y)$ is a vector of $n^2$ terms, and the ${(ij)}^{th}$ term is given by $e^Y_{ij}\cdot \chi_{ij}$. Here $\chi_{ij}$ is the permanent of submatrix obtained from $e^Y$  after removing row and columns $i,j$. Therefore each term in $\nabla f(Y)$ can be approximated with $\gamma$-approximation using FPTAS in \cite{sinclair} for approximating the permanent of a non-negative matrix. 
\paragraph{Running time}
Observe that $f_{ij}(Y)=\mbox{perm}(e^{Y'_{ij}})$, where $Y'_{ij}$ denotes the matrix obtained from $Y$ after removing the row $i$ and column $j$. Thus, $(1\pm \delta)$ approximation to $f(Y)$ can be obtained in time polynomial in $n,1/\delta$ using the FPTAS given in \cite{sinclair} for computing permanent of a non-negative matrix. Since, $1/\delta$ is polynomial in $n,1/\epsilon, 1/q_{min}$, this gives polynomial running time for estimating $f(Y)$. Similar observations hold for estimating the gradient $\nabla f_{ij}(Y)$ in above.

\paragraph{Correctness} The correctness of the above algorithm is established by the following two lemmas:
\begin{lemma} \label{C4} If $\widehat{f}(Y) \le (1+3\delta)Q$ and all the other constraints are satisfied, then $Y \in K_\epsilon$. \end{lemma}
\begin{proof}See the appendix. \end{proof}
%If $\widehat{f(Y)} > 1+3\delta$ then $f(Y)>1+\delta$. 
%We obtain a separating hyperplane for such $Y$. Set $C=\widehat{\nabla}{f}(Y)$, where $\widehat{\nabla}{f}(Y)$ is a $\gamma$ approximation of $\nabla f(Y)$. That is,
%$$ \nabla f(Y) (1-\gamma) \le \widehat{\nabla} f(Y) \le \nabla f(Y)(1+\gamma)$$
%  We will fix $\gamma$ later. Observe that the gradient of $f(Y)$
%is a vector of $n^2$ terms, where the ${(ij)}^{th}$ term is given by $e^Y_{ij}\cdot \chi_{ij}$, where $\chi_{ij}$ is the permanent of submatrix obtained from $e^Y$  after removing row and columns $i,j$. Therefore each term in $\nabla f(Y)$
%can be approximated with $\gamma$-approximation using FPTAS in \cite{sinclair} for approximating the permanent of a non-negative matrix. 
\begin{lemma}\label{C5}
Suppose the $ij^{th}$ constraint is violated, i.e., $\widehat{f}_{ij}(Y) > (1+3\delta)Q_{ij}$.
%Given $\frac{-n\log n}{q_{min}}e \le Y \le 0$, such that $g(Y)> (1+\delta)q$. 
Then,  $C=\widehat{\nabla}f_{ij}(Y)$ gives a separating hyperplane for $\textbf K$, that is, $C\bullet(X - Y)\le 0, \forall X\in \textbf{K}$.
\end{lemma}
\begin{proof} See the appendix. The proof uses the bounds on $X,Y$ and $\nabla f_{ij}(Y)$ established in Lemma \ref{C2} and Lemma \ref{C3}, respectively.
\end{proof}

%Let $OPT$ be the value of optimal solution. 
%----------------------------------------------------------
\begin{theorem}
Using the separating oracle given by Lemma \ref{SEP} with the ellipsoid method, a distribution $\{p_\sigma\}$ over permutations can be constructed in time $poly(n, \frac{1}{\epsilon}, \frac{1}{q_{min}})$, such that 
\begin{itemize*}
\item $(1-\epsilon)Q \le \sum_\sigma p_\sigma M_\sigma \le Q$
\item $p$ has close to maximum entropy, i.e., $\sum_\sigma p_\sigma \log{p_\sigma} \le (1-\epsilon)OPT_E$, where $OPT_E (\le 0)$ is the optimal value of $(\ref{max-entropy})$.
\end{itemize*}
\end{theorem}
\begin{proof}
Using the above separating oracle with the ellipsoid method \cite{grotschel}, after polynomial number of iterations we will either get a solution $Y\in \textbf{K}_\epsilon(t)$, or declare that there is no feasible solution. Thus, by binary search over the $t$, we can get a solution $\bar{Y}$ such that $\bar{Y}\in K_\epsilon$ and $Q\bullet \bar{Y}-1 \ge OPT$. 
The dual solution thus obtained will have an objective value equal to or better than optimal but may be infeasible. We reduce each of the $\bar{Y}_{ij}$s by a small 
amount ($\frac{1}{n} \log({1+\epsilon})$) to construct a new feasible but sub-optimal solution $\widehat{Y}$. 
Some simple algebraic manipulations show that the new solution $\widehat{Y}$ satisfies:
$(1-\epsilon)Q \le \sum_\sigma e^{\widehat{Y}\bullet M_\sigma} M_\sigma \le Q$
Thus, $\widehat{Y}$ is a feasible solution to the dual, and, $Q\bullet \widehat{Y} -1 \le OPT$.
We can now construct the distribution $p_\sigma$ as $p_\sigma = e^{\widehat{Y}\bullet M_\sigma}$. 
Then from above, $(1-\epsilon)Q \le \sum_\sigma p_\sigma M_\sigma \le Q$.
Also,
\vspace{-0.05in}
\begin{equation}
\begin{array}{rcl}
\sum_\sigma p_\sigma \log{p_\sigma} &  = & \sum_\sigma e^{\widehat{Y}\bullet M_\sigma} M_\sigma \bullet \widehat{Y} \\ 
 & \le & \hspace{0.1in} (1-\epsilon)Q\bullet \widehat{Y}\\ %- (1-\epsilon)\nonumber\\
 & \le & \hspace{0.1in} {(1-\epsilon)(OPT+1)}\nonumber\\
 & = & \hspace{0.1in} {(1-\epsilon)OPT_E}\nonumber
\end{array}
\end{equation}
\end{proof}
\section{Conclusion}
We introduced a Proportional Betting mechanism for permutation betting which can be readily implemented by solving a convex program of polynomial size. 
%under parimutuel call auction model. 
More importantly, the mechanism was shown to admit an efficient parimutuel pricing scheme, wherein only $n^2$ marginal prices were needed to price the bets. Further, we demonstrated that these marginal prices can be used to construct meaningful joint distributions over the exponential outcome space. 
%We proposed a tractable proportional betting mechanism and solved the parimutuel pricing problem for permutation betting under a parimutuel call auction model. We showed that  a marginal price vector of size $n^2$ is is unique and computable by solving a convex program of polynomial size.  Further, we presented a maximum entropy approach to obtain the complete price vector over the exponential outcome space using these marginal prices.
%There are several avenues of further research here. An interesting direction of research would be to rigorously establish the connection between traders' belief about the outcome distribution and the state prices for this model. Empirical studies exploring this aspect also hold interest. Another open question is whether there exist more generalized betting languages which are tractable and provide further information about the outcome distributions.

The proposed proportional betting mechanism was developed by relaxing a `fixed reward betting mechanism'. An interesting question raised by this work is whether the fixed betting mechanism could provide further information about the outcome distribution.
Or, in general, how does the complexity of the betting language relates to the information collected from the market? A positive answer to this question would justify exploring approximation algorithms for the more complex fixed reward betting mechanism.
%That would provide   kind of approximation algorithms are possible for this mechanism.
%`Boolean game' proposed in this paper and provide further information about the outcome distributions.
\paragraph{Acknowledgements}
We thank Arash Asadpour and Erick Delage for valuable insights and discussions.
%\shcomment{TODO: open questions}
%There are several avenues of open research here. 
%%%%%%%%%%%%%%%%%%%%%%%%%%%%%%%%%%%%%%%%%%%%%%%%%%%%%%%%%%%%%%%%%%%%%%%%%%%%%%%%%%%%%%%%%%%%5

% do the bibliography:
\bibliographystyle{plain}
\bibliography{paper}
% where ``my-bibliography-file.bib'' is the name of the file with all the 
% BibTeX entries.

%%%%%%%%%%%%%%%%%%%%%%%%%%%%%%%%%%%%%%%%%%%%%%%%%%%%%%%%%%%%%%%%%%%%%%%%%%%%%%%%%%%%%%%%%%%%5
\appendix
\begin{center}
    {\bf APPENDIX}
  \end{center}
%----------------------------------------------------------
\paragraph{Proof of Theorem \ref{T1}}
Consider the complete bipartite graph with the $n$ candidates in one set and the $n$ positions in the other set. In our betting mechanism, each bidder $k$ bids on a subset of edges in this graph which is given by the non-zero entries in his bidding matrix $A_k$. A bidder is ``satisfied" by a matching (or permutation) in this graph if at least one of the edges he bid on occurs in the matching.
The separation problem  for the linear program in (\ref{0-1game}) corresponds to finding the matching that satisfies the maximum number of bidders.
Thus, it can be equivalently stated as the following matching problem:\\

\noindent {\it Matching problem:} Given a complete bipartite graph $K_{n,n}=(V_1, V_2, E)$, and a collection ${\cal C}=\{E_1, E_2, \ldots, E_m\}$ of $m$ subsets of $E$. %Let $E_k \subseteq E$ be the set of edges that bidder $k$ is bidding on. 
Find the perfect matching $M\subset E$ that intersects with maximum number of subsets in $\cal C$. \\

\noindent {\it MAX-SAT problem:} Given a boolean formula in CNF form, determine an assignment of $\{0,1\}$ to the variables in the formula that satisfies the maximum number of clauses.\\

\noindent {\it Reduction from MAX-SAT to our matching problem:} Given the boolean formula in MAX-SAT problem with $n$ variables $x, y, z \ldots$. Construct a complete bipartite graph $K_{2n,2n}$ as follows. For each variable $x$, add two nodes $x$ and $x'$ to the graph. And, for the possible values $0$ and $1$ of $x$, construct two nodes $x_0$ and $x_1$. Connect by edges all the nodes corresponding to the variables to all the nodes corresponding to the values. Now, create the collection $\cal C$ as follows. For $k^{th}$ clause in the boolean formula, create a set $E_k$ in $\cal C$. For each negated variable $x$ in this clause, add edge $(x,x_0)$ to $E_k$; and for each non-negated variable $x$ in the clause, add an edge $(x,x_1)$ to $E_k$. 

We show that every solution of size $l$ for the MAX-SAT instance corresponds to a solution of size $l$ for the constructed matching problem instance and vice-versa.
Let there is an assignment that satisfies $l$ clauses of MAX-SAT instance. Output a matching $M$ in the graph $K$ as follows. For each variable $x$, consider the nodes $x, x', x_0, x_1$. Let the variable $x$ is assigned value $0$ in the MAX-SAT solution. Then, add edges $(x,x_0)$, $(x',x_1)$ to $M$. Otherwise, add edges $(x,x_1)$, $(x',x_0)$ to $M$. 
%add edges Connect $x$ to the node corresponding to the value assigned to $x$ in MAX-SAT solution $(x_0 or x_1)$ and $x'$ to the other node. 
It is easy to see that the resulting set $M$ is a matching.
Also, if a clause $k$ satisfied in the MAX-SAT problem, then the matching $M$ will have an edge common with $E_k$. Therefore $M$ intersects with at least $l$ subsets in $\cal C$.

Similarly, consider a solution $M$ to the matching problem. Form a solution to the MAX-SAT problem as follows. Let the set $E_k$ is satisfied (intersects with $M$). Then, one of the edges in $E_k$ must be present in $M$. Let $(x,x0)$ ($(x,x1)$) is such an edge. Then, assign $0$ ($1$) to $x$. Because the $M$ is a matching, any node $x$ will have at the most one edge in $M$ incident on it, and both $(x,x_0)$ and $(x,x_1)$ cannot be present $M$. This ensures that takes $x$ will take at the most one value $0$ or $1$ in the constructed assignment. For the remaining variables, assign values randomly. By construction, if a set $E_k$ is satisfied in the matching solution, the corresponding $k^{th}$ clause must be satisfied in the MAX-SAT problem - resulting in a solution of size at least $l$ to MAX-SAT. This completes the reduction.

Note that in above, if we reduced from MAX-2-SAT, then each subset $E_k$ would contain exactly two edges, that is, we would get an instance in which each bidder bids on exactly two candidate-position pairs. Because MAX-2-SAT is NP-hard, this proves that this special case is also NP-hard. %(2-SAT is polynomial time solvable but MAX-2-SAT is NP-hard)

%----------------------------------------------------------

\paragraph{Proof of Lemma \ref{L4}}

The dual for (\ref{primal-market}) is:
\begin{equation}\label{L4-1}
\begin{array}{lll}
\displaystyle\min_{y,Q} & q^{T}y &\\
\textrm{s.t.} & A_k\bullet Q+y_k\geq \pi_k & \forall k\\
& Qe=e &\\
& Q^{T}e=e &\\
& y\geq 0 &\\
& Q\geq 0 &\\
\end{array}
\end{equation}
The dual for (\ref{basicProblem}) is:
\begin{equation}\label{L4-2}
\begin{array}{lll}
\displaystyle\min_{y, p} & q^Ty & \\
\mbox{s.t.} & \sum_{\sigma} (A_{k} \bullet M_\sigma) p_\sigma + y_k \ge \pi_k & \forall k \\
& \sum_\sigma p_\sigma =1 & \\
& y\ge 0 & 
\end{array}
\end{equation}
Suppose $p'_\sigma$ is a solution to (\ref{L4-2}), and $\sum_\sigma p'_\sigma M_\sigma = Q'$, then the first constraint in (\ref{L4-2}) is equivalent to $A_k \bullet Q' + y_k \ge \pi_k$. Hence, for any solution $p'_\sigma$ to (\ref{L4-2}), there is a corresponding solution $Q'=\sum_\sigma p'_\sigma M_\sigma$ to (\ref{L4-1}) with the same objective value. Thus, if $Q$ is an optimal solution to (\ref{L4-1}), then all $\{p_\sigma\}$ satisfying $\sum_\sigma p_\sigma M_\sigma =Q$ have the same objective value and are optimal. 
%if $Q$ is optimal solution to (\ref{L4-2}), then for any $\{p'_\sigma\}$
%----------------------------------------------------------
\paragraph{Proof of Lemma \ref{Hardness-entropy}}
The optimality condition for the dual problem in (\ref{dual-max-entropy}) is specified as (setting derivative to $0$): 
\begin{equation}
\begin{array}{rcl}
Q  & = & \sum_\sigma e^{Y\bullet M_\sigma-1}M_\sigma\\
\end{array}
\end{equation}
Thus, given a certificate $Y$, verifying its optimality requires computing the function $f(Y)=\sum e^{Y\bullet M_\sigma}M_\sigma$ for a given $Y$. 
%We show that this verification problem is \#P-hard using a reduction from the $\#P$-hard problem of computing permanent of a $(0,1)$ matrix. This will prove that the optimization problem is NP-hard.
Note that the $ij^{th}$ component of this function is given by 
	$$e^{Y_{ij}} \sum_{\sigma:j=\sigma(i)} e^{Y\bullet M_\sigma} = e^{Y_{ij}} \mbox{perm}(e^{Y'})$$ 
where $Y'$ is the matrix obtained from $Y$ after removing row $i$ and column $j$. We show that computing the permanent of 
$e^{Y'}$ is \#P-hard by reducing it to the %\#P-hard 
problem of computing permanent of a $(0,1)$ matrix. % \cite{sinclair}. 
The reduction uses the technique from the proof of Theorem $1$ in \cite{chen08}. We repeat the construction below for completeness. Suppose $A$ is a $(n-1) \times (n-1)$ $(0,1)$ matrix whose permanent we wish to find. Then, construct a matrix $Y'$ as follows:
		$$	Y'_{kl} = \left\{\begin{array} {ll}
					\log(n!+2) & A_{kl}=1 \\
					\log(n!+1) & A_{kl}=0 		
					\end{array}
					\right.
		$$
Then, $\mbox{perm}(e^{Y'}) \mod (n!+1)= \mbox{perm}(A) \mod (n!+1) = \mbox{perm}(A)$, since $\mbox{perm}(A)\le n!$. Hence, even the verification problem for this optimization problem is at least as hard as computing the permanent of a $(0,1)$-matrix. 

%----------------------------------------------------------
\paragraph{Proof of Lemma \ref{C1}} 

Observe that the problem in (\ref{max-entropy}) involves implicit constraints $\sum{p_\sigma}=1$.
Further, we show that the equality constraints can be relaxed to inequality.
We will show that it is impossible that $\sum{p_\sigma M_\sigma}<Q$
for some elements in the optimal solution. Observe that the matrix 
$Q-\sum{p_\sigma M_\sigma}$ has the property that each row and each column 
sums up to $1-\sum{p_\sigma}$. That is, $(Q-\sum{p_\sigma M_\sigma})/(1-\sum{p_\sigma})$ is a doubly-stochastic matrix. {\it Birkhoff-von Neumann theorem} \cite{birkhoff} proves 
that any doubly stochastic matrix can be represented as a convex combination of permutation matrices.  Since $Q-\sum{p_\sigma M_\sigma}>0$, there must be at least one strictly positive coefficient in the Birkhoff-von Neumann decomposition of this matrix. This means that we can increase at least one $p_\sigma$ a little bit without
violating the inequality constraint. 
However, the derivative of the objective
w.r.t one variable is $\log{p_\sigma}$. Therefore, when $p_\sigma<1$
, increasing $p_\sigma$ will always decrease the objective value,
which contradicts with the assumption that we have already reached the 
optimal. Thus we have shown that the problem in (\ref{primal-entropy}) shares the same
optimal solution as (\ref{max-entropy}).
%----------------------------------------------------------

\paragraph{Proof of Lemma \ref{C2}}
Note that $Y=-\log{n}\times
ones(n,n)$ forms a feasible solution to (\ref{dual-tmp}). 
Hence, the optimal value to the dual must be greater
than $-n\log{n}-1$, that is, 
$0 \ge OPT \ge -n\log{n}-1$.
%thus we have that at optimal $Q\bullet Y\geq -n\log{n}-1$. 
Also, from KKT conditions, the optimal solutions to the primal and dual are related as \mbox{${p_\sigma} = e^{Y\bullet M_\sigma}$}.
Hence, as discussed in proof of Lemma \ref{C1} for the primal solution, the optimal dual solution must satisfy \mbox{$\sum{e^{Y\bullet M_\sigma} M_\sigma}  = Q$}, implicitly leading to \mbox{$\sum{e^{Y\bullet M_\sigma}}=1$} at optimality. Along with the lower bound on $OPT$, this gives $Q\bullet Y \ge -n\log n$, which implies
%Now, $q_{min} = \min\{Q_{ij}\} | Q_{ij} \ne 0\}$. For $ij$ such that $Q_{ij}=0$, we can set $Y_{ij}=-\infty$ and exclude them from the formula. For the rest, 
%Now, using the fact that at optimal $\sum{e^{Y\bullet M_\sigma}}=1$, $Q\bullet Y\ge -n\log{n}$, which implies:
$Y_{ij} \ge -n\log{n}/q_{min}$. 
%Define $A=\{(i,j)|Q(i,j)\neq0\}$,$q_{min}=min\{Q(i,j)|(i,j)\in A\}$, we know that $Y_{ij}\geq(-n\log{n}-1)/q_{min}$. 

%----------------------------------------------------------

\paragraph{Proof of Lemma \ref{C3}}
%Consider the ${ij}^{th}$ component $g(Y)=f_{ij}(Y)=\sum_\sigma e^{Y\bullet P}M_{\sigma,ij}$.
The gradient of $g(Y)$ is $\nabla g(Y)=\sum_\sigma (e^{Y\bullet M_\sigma}M_{\sigma,ij}) M_\sigma$. That is, $\nabla g(Y)$ is an $n\times n$ matrix defined as:
$${\textstyle
\begin{array}{l}
\nabla g(Y)_{k,l} = \vspace{0.1in}\\
\left\{\begin{array}{ll}
%\nabla g(Y)_{kl}  =  \left\{\begin{array}{ll}
		{\textstyle \sum_{\sigma: j=\sigma(i)} {e^{Y\bullet M_\sigma}}} & {\scriptstyle \mbox{if } (k, l) = (i,j)} \\
		{\textstyle \sum_{\sigma: j=\sigma(i),l=\sigma(k)} {e^{Y\bullet M_\sigma}}} & {\scriptstyle \mbox{if } \{k,l\} \bigcap \{i,j\} = \phi} \\
		{\textstyle 0} & {\scriptstyle \mbox{o.w., if } \{k,l\} \bigcap \{i,j\} \ne \phi}\\
		\end{array}
		\right.
\end{array}
}
$$
%Note that $f(Y)$ gives the ``permanent" of the non-negative matrix $e^Y$.
%Here we use the notation $e^Y$, where $Y$ is a matrix, to mean componentwise exponentiation: $(e^Y)_{ij} = e^{Y_{ij}}$ . 
%Denote by $\chi_{ij}$, the permanent of the submatrix obtained after removing row $i$ and column $j$ from $e^Y$. Then, 
%The above $n \times n$ matrix can also be represented as:
We will use the notation $e^Y$, where $Y$ is a matrix, to mean
component-wise exponentiation: $(e^Y)_{ij} = e^{Y_{ij}}$ . Let $\chi$ denote the permanent of the non-negative matrix $e^Y$. Denote by $\chi_{ij}$, the ``permanent" of the submatrix obtained after removing row $i$ and column $j$ from $e^Y$.
Then, observe that $g(Y)=e^{Y_{ij}}\cdot\chi_{ij}$.
Also, the gradient of $g(Y)$ can be written as:
$$
\begin{array}{l}
\nabla g(Y)_{k,l} = \vspace{0.1in}\\
\hspace{0.3in}  \left\{\begin{array}{ll}
		e^{Y_{ij}}\cdot \chi_{ij} & {\scriptstyle \mbox{if } (k,l) = (i,j)} \\
		e^{Y_{ij}}e^{Y_{kl}} \cdot \chi_{ij,kl} & {\scriptstyle \mbox{if } \{k,l\} \bigcap \{i,j\} = \phi}\\
		0 & {\scriptstyle \mbox{o.w., if } \{k,l\} \bigcap \{i,j\} \ne \phi}\\
		\end{array}
		\right.
\end{array}
$$
where $\chi_{ij,kl}$ denotes the permanent of the matrix obtained after removing rows $i,k$ and columns $j,l$ from $e^Y$. 
%Again, $\chi_{ij,kl}$ can be computed to $(1+\epsilon)$-approximation using the FPTAS for approximating the permanent of a non-negative matrix.  

Using the relation between permanent of a matrix and its submatrices, observe that:
\begin{eqnarray}
||\nabla g(Y)||_1 & = & e^{Y_{ij}} \cdot \chi_{ij} + e^{Y_{ij}}\sum_{kl: ij \ne kl} e^{Y_{kl}}\cdot \chi_{ij,kl} \nonumber\\
		& = & n e^{Y_{ij}} \chi_{ij} \nonumber\\
		& = & ng(Y)\nonumber
\end{eqnarray}
Hence,
\begin{eqnarray}
||\nabla g(Y)||_2 \le &  ng(Y) & \le  \sqrt{n} ||\nabla g(Y)||_2 \nonumber
\end{eqnarray}
%----------------------------------------------------------
\paragraph{Proof of Lemma \ref{C4}}
%It is easy to see that if $\widehat{f}(Y) \le (1+3\delta)Q$, then $Y \in K_\epsilon$. 
For any such $Y$,
$$f(Y) \le \frac{\widehat{f}(Y)}{(1-\delta)} \le \frac{(1+3\delta)}{(1-\delta)}Q %\le (1+3\delta)(1+2\delta) Q
\le (1+12\delta) Q\le (1+\epsilon)Q$$

%----------------------------------------------------------
\paragraph{Proof of Lemma \ref{C5}}
%Next, we show the correctness of separating hyperplane construction. 
Suppose the $ij^{th}$ constraint is violated. That is, $\widehat{f}_{ij}(Y) > (1+3\delta)Q_{ij}$. This implies that $f_{ij}(Y) >(1+\delta)Q_{ij}$. This is because if $f_{ij}(Y) \le (1+\delta)Q_{ij}$, then $\widehat{f}_{ij}(Y) \le (1+\delta)f(Y) \le (1+3\delta)Q_{ij}$. 

In below we denote the function $f_{ij}(Y)$ by $g(Y)$ and $Q_{ij}$ by $b$.
%, and prove the following lemma about the correctness of separating hyperplane construction.
Given any $X\in K$, %and $Y\notin K$ such that $g(Y) \ge 1+\delta$.
since $g(\cdot)$ is a convex function, 
$$\nabla g(Y)^T (X-Y) \le g(X)-g(Y) \le b-g(Y)$$ % \le \min\{-\delta, 1-\frac{1}{n}||\nabla f(Y)||\}$$
Therefore, using the bounds on $X$ and $Y$,
%\begin{equation}
%\begin{array}{l}
%\end{array}
%\end{equation}
\begin{equation}
\begin{array}{l}
\widehat{\nabla} g(Y)^T(X-Y) \nonumber\\
\hspace{0.2in} \le {\textstyle \nabla g(Y)^T (X-Y) + ||\nabla g(Y)-\widehat{\nabla} g(Y)||\cdot||X-Y||} \nonumber\\
\hspace{0.2in} \le b-g(Y) + \gamma ||\nabla g(Y)||\frac{n^2\log{n}}{q_{min}} \nonumber\\
\hspace{0.2in} \le b-g(Y) + \gamma \cdot ng(Y) \cdot \frac{n^2\log{n}}{q_{min}} \nonumber\\
\hspace{0.2in} \le b-b(1+\delta)( 1 - \gamma \frac{n^3\log{n}}{q_{min}}) \nonumber
\end{array}
\end{equation}
where the second last inequality follows from the bound on gradient given by Lemma \ref{C3}. 
%(\ref{boundGrad}). 
The last inequality follows from the observation made earlier that $g(Y) = f_{ij}(Y) > 1 + \delta$. Now,
%Now set $\gamma$ so that $(1+\delta)( 1 - \gamma \frac{n^3\log{n}}{q_{min}}) \ge 1$. That is,
$$ \gamma = \frac{\delta q_{min}}{2n^4}\le \frac{\delta}{1+\delta}\cdot\frac{q_{min}}{n^3\log{n}}$$
Hence, from above,  
$$\widehat{\nabla} g(Y)^T(X-Y) \le 0$$
%----------------------------------------------------------

\end{document}